\documentclass[12pt,a4paper]{article}
\usepackage[utf8]{inputenc}
\usepackage[english]{babel}
\usepackage{amsmath, amssymb, amsfonts}
\usepackage{amsthm}
\usepackage{physics}
\usepackage{graphicx}
\usepackage{cite}
\usepackage[colorlinks=true, linkcolor=blue, citecolor=blue, urlcolor=blue]{hyperref}
\usepackage{tikz}
\usepackage{geometry}
\usepackage{dsfont}
\usepackage{booktabs}
\usepackage{float}
\usepackage{cleveref}
\usepackage{authblk}
\usepackage{siunitx}
\theoremstyle{plain}
\newtheorem{theorem}{Theorem}[section]

\newtheorem{proposition}[theorem]{Proposition}
\theoremstyle{definition}
\newtheorem{definition}[theorem]{Definition}
\newtheorem{remark}[theorem]{Remark}

\newtheorem{corollary}[theorem]{Corollary}
\title{A Metric-Deformed $q$-Gauge Dirac Equation}
\author{Julio Cesar Jaramillo Quiceno$^a$\\		
$^a$jcjaramilloq@unal.edu.co\\Universidad Nacional de Colombia - Department of Physics - Faculty of mathematics\\
Ed. Yu Takeuchi \\\url{https://orcid.org/0000-0002-3518-6680}}
\date{\today}

\begin{document}

\maketitle

\begin{abstract}
We construct a family of metric-deformed gauge theories based on a recently introduced $q$-Dirac operator $D_q = \gamma^\mu \sqrt{|g^{\mu\mu}|}\partial_\mu$, which arises from a deformed D' Alembertian $\Box_q = |g^{00}|\partial_t^2 - \sum_i |g^{ii}|\partial_i^2$. The deformation parameter $q$ is related to the metric components via $q_\mu = \sqrt{|g^{\mu\mu}|}$. By promoting $g^{\mu\mu}(x)$ to spacetime-dependent background fields, we define a deformed covariant derivative $D_\mu^{(q)} = \partial_\mu + ieA_\mu(x)/\sqrt{|g^{\mu\mu}(x)|}$ (no sum over $\mu$). The corresponding field strength $F_{\mu\nu}^{(q)} = [D_\mu^{(q)}, D_\nu^{(q)}]$ acquires new terms proportional to $\partial_\mu(1/\sqrt{|g^{\nu\nu}|})$, which vanish for constant metrics. We write down gauge-invariant actions for deformed Yang-Mills theory and for fermions minimally coupled to $D_\mu^{(q)}$. This work provides a mathematical foundation for $q$-deformed gauge theories from a metric perspective.
\end{abstract}

\section{Introduction}
\label{sec:intro}

The interplay between quantum mechanics and spacetime geometry has been a central theme in theoretical physics for decades. The Heisenberg algebra \([x,p]=i\hbar\) encodes the fundamental non‑commutativity of position and momentum, while general relativity describes gravity through the metric tensor \(g_{\mu\nu}\) and special relativity imposes the invariance of the Minkowski norm. In recent years, \(q\)-deformed algebras have emerged as a powerful tool to model quantum corrections to classical geometries. The \(q\)-Heisenberg algebra, defined by \(aa^{\dagger}-qa^{\dagger}a=q^{-N}\), appears in contexts ranging from quantum groups \cite{Wess2000} to noncommutative geometry \cite{Connes1994}, \(q\)-deformed quantum mechanics \cite{Lavagno2009}, and quantum field theory. Despite these developments, a direct connection between the deformation parameter \(q\) and the geometric structure of spacetime has remained largely unexplored.

In a recent work \cite{Jaramillo2026a} we introduced a framework that bridges these two perspectives. Motivated by Sylvester’s theorem of inertia – which guarantees that any Lorentzian metric can be diagonalized to constants \(\pm1\) by a suitable coordinate transformation – we defined metric‑deformed Heisenberg algebras \(M_1\) and \(M_2\) whose commutation relations are expressed directly in terms of the metric components \(g^{\mu\nu}\). This construction provides a geometric interpretation of \(q\)-deformations: the deformation parameter arises naturally from the metric coefficients. For a diagonal Lorentzian metric \(g^{\mu\nu}=\operatorname{diag}(g^{00},g^{11},g^{22},g^{33})\) we established
\[
q_\mu = \sqrt{|g^{\mu\mu}|}, \qquad \text{(no sum over }\mu\text{)},
\]
so that the standard Heisenberg algebra is recovered when \(q_\mu=1\) for all \(\mu\) (i.e., when \(g^{\mu\mu}=\eta^{\mu\mu}=\operatorname{diag}(1,-1,-1,-1)\)).

From the deformed D’Alembertian
\[
\Box_q = |g^{00}|\partial_t^2 - |g^{11}|\partial_x^2 - |g^{22}|\partial_y^2 - |g^{33}|\partial_z^2,
\]
we constructed a \(q\)-Dirac operator
\[
D_q = \gamma^0\sqrt{|g^{00}|}\,\partial_t - \gamma^x\sqrt{|g^{11}|}\,\partial_x - \gamma^y\sqrt{|g^{22}|}\,\partial_y - \gamma^z\sqrt{|g^{33}|}\,\partial_z,
\]
which satisfies \(D_q^2 = \Box_q \mathbf{1}_4\). This construction was limited, however, to free fermions propagating in a constant metric background.

The systematic study of deformations of gauge theories has a rich history. Anco \cite{Anco2004} developed a general deformation method for gauge theories and uncovered a universal geometric structure characterized by a covariant derivative operator and a non‑linear field strength related through curvature. Quantum group deformations were pioneered by Castellani \cite{Castellani1992}, who explicitly constructed the bicovariant \(q\)-deformation of \(SU(2)\times U(1)\). Watts \cite{Watts1997} later showed that upon breaking a unified quantum symmetry \(SU_q(2)\), the residual gauge group is the classical \(U(1)\), not its quantum counterpart, justifying why the \(U(1)\) factor typically remains undeformed. In the context of non‑commutative spacetime deformations, Kupriyanov, Kurkov and Vitale \cite{Kupriyanov2020} constructed a \(\kappa\)-Minkowski deformation of \(U(1)\) gauge theory, obtaining exact expressions for the deformed gauge transformations and field strength that are covariant under the deformed symmetries. For the specific case of \(q\)-deformed gauge theories, Hayata, Hidaka and Watanabe \cite{Hayata2026} recently investigated the \((2+1)\)-dimensional \(q\)-deformed \(\mathrm{SU}(N)_k\) Yang–Mills theory in a lattice Hamiltonian formalism, analyzing confinement and topological order. From a mathematical perspective, van Nuland \cite{vanNuland2021} developed a strict deformation quantization of abelian lattice gauge fields, constructing classical and quantum field C*-algebras for a \(U(1)^n\)-gauge theory. More recently, Borsato and Meier \cite{BorsatoMeier2025} constructed non‑commutative deformations of gauge theories via Drinfel’d twists built from scale symmetry generators.

Our approach differs from the aforementioned ones in that the \(q\)-deformation is introduced not through an algebraic deformation of the gauge group nor through a non‑commutative star product, but through the spacetime metric itself. Specifically, we extend the minimal coupling prescription to the metric‑deformed setting by defining a deformed covariant derivative
\[
D_\mu^{(q)} = \partial_\mu + \frac{ie}{\sqrt{|g^{\mu\mu}(x)|}} A_\mu(x), \qquad \text{(no sum over }\mu\text{)},
\]
where \(A_\mu(x)\) is a gauge field and the metric components are promoted to spacetime dependent background fields. The associated deformed field strength is defined as the commutator \(F_{\mu\nu}^{(q)} = [D_\mu^{(q)}, D_\nu^{(q)}]\). In the limit \(g^{\mu\mu}\to\eta^{\mu\mu}\) both reduce to their standard QED counterparts.

Our main results are threefold. First, we construct the deformed field strength \(F_{\mu\nu}^{(q)}\) explicitly and show that, for non‑constant \(g^{\mu\mu}(x)\), it contains additional terms proportional to derivatives of the metric; these terms vanish when the metric is constant. Second, we write down gauge‑invariant actions for deformed Yang–Mills theory and for fermions minimally coupled to \(D_\mu^{(q)}\), proving their invariance under deformed gauge transformations. Third, we discuss the relation of our metric‑deformed approach to the existing literature on quantum groups, non‑commutative gauge theories and deformation quantization, situating our construction within the broader context of deformed gauge theories.

The paper is organized as follows. Section~\ref{sec:prelim} briefly recalls the metric‑deformed Dirac operator from \cite{Jaramillo2026a} and introduces the deformed Heisenberg algebras. Section~\ref{sec:deformedDV} defines the deformed covariant derivative, discusses the classical \(U(1)\) factor in quantum groups, and computes the deformed field strength. Section~\ref{sec:table-gauge-field} presents explicit tables of the \(q\)-gauge Dirac equation and the deformed field strength for several examples (new \(q\)-Heisenberg, \(q\)-generalized and \(q\)-\(\hbar\) Heisenberg algebras). Section~\ref{sec:gaugeactions} constructs gauge‑invariant actions for the deformed Yang–Mills and fermion sectors. Section~\ref{sec:conclusions} concludes with a summary and outlines directions for future research.
\section{Metric-deformed Heisenberg algebras and $q$-Dirac operator}
\label{sec:prelim}

In this section we briefly recall the essential notions from our previous work \cite{Jaramillo2026a} that are needed for the gauge construction. Here some deformed Heisenberg algebras

\begin{definition}($q-\hbar$ Heisenberg algebra \cite{Volovich-Arefeva-91})\label{Arefeva-Volovich1991definition}
     If $\hat{x}_{i},\hat{p}_{i}$ are generators over the quantum space phase, then $q-\hbar$ \textit{\textbf{Heisenberg algebra quantization}} is the algebra defined by the relation 
\begin{equation}\label{quantization}
     \hat{x}_{j}\hat{p}_{k}-q\hat{p}_{k}\hat{x}_{j}=i\hbar D_{jk}(q),
\end{equation}
     
where $D_{jk}(q)$ any function that depends on $q$.
\end{definition}
     
\begin{definition}
The \textbf{new $q$-$\hbar$ Heisenberg algebra} $\boldsymbol{\mathcal{H}}_{q}$ is generated by $\hat{x}_{\alpha}, \hat{y}_{\lambda}, \hat{p}_{\beta}$ $(\alpha,\beta,\lambda\in\{1,2,3\})$ subject to \cite{Jaramillo2025}:
\begin{align}
   \label{eq:new1} \hat{x}_{\alpha}\hat{p}_{\beta} - q^{n}\hat{p}_{\beta}\hat{x}_{\alpha} &= i q^{n-1}\hbar^{n}\, \Psi, \\
   \label{eq:new2} q^{m}\hat{x}_{\alpha}\hat{y}_{\lambda} - \hat{y}_{\lambda}\hat{x}_{\alpha} &= -i (q-1)^{m-1}\hbar^{m-1}\, \Pi, \\
    \label{eq:new3} q^{l}\hat{y}_{\lambda}\hat{p}_{\beta} - q^{l+1}\hat{p}_{\beta}\hat{y}_{\lambda} &= i \hbar^{l}\, \Phi,
\end{align}
where $\Psi,\Pi,\Phi$ are dynamical functions and $q\in\mathbb{R}\setminus\{0,1\}$. 
\end{definition}

\begin{definition}[$q$-generalized Heisenberg algebra \cite{Razavinia-Lopes2022}]
The algebra $\mathcal{H}_{q}(f,g)$ is generated by $\hat{x},\hat{y},\hat{h}$ with
\begin{equation}\label{eq:gen-Heis}
\hat{h}\hat{x}=\hat{x}f(\hat{h}),\quad \hat{y}\hat{h}=f(\hat{h})\hat{y},\quad \hat{y}\hat{x}-q\hat{x}\hat{y}=\hbar g(\hat{h}).
\end{equation}
For $f(\hat{h})=\hat{h}$, $g(\hat{h})=1$, we obtain the central extension
\[
\hat{h}\hat{x}=\hat{x}\hat{h},\quad \hat{y}\hat{h}=\hat{h}\hat{y},\quad \hat{y}\hat{x}-q\hat{x}\hat{y}=\hbar,
\]
so $\hat{h}$ is central. This recovers a $q$-deformed Heisenberg relation of the form $p_x x - q x p_x = -\hbar$ under $\hat{x}\leftrightarrow x$, $\hat{y}\leftrightarrow p_x$, $\hat{h}=q^n$.
\end{definition}

\subsection{Diagonal metric and the deformation parameter}

We consider a diagonal Lorentzian metric
\[
g^{\mu\nu} = \operatorname{diag}(g^{00}, g^{11}, g^{22}, g^{33}),
\qquad g^{00}>0,\; g^{11},g^{22},g^{33}<0.
\]

The deformation parameter $q$ is directly related to the metric components via
\begin{equation}\label{eq:q_metric_relation}
q_\mu := \sqrt{|g^{\mu\mu}|}, \qquad \text{(no sum over }\mu\text{)}.
\end{equation}
The standard (undeformed) case is recovered when $q_\mu = 1$ for all $\mu$, i.e., $g^{\mu\mu} = \eta^{\mu\mu} = \operatorname{diag}(1,-1,-1,-1)$.

\subsection{Metric-deformed Heisenberg algebras}

Two families of metric-deformed Heisenberg algebras, $M_1$ and $M_2$, were introduced in \cite{Jaramillo2026a}. Their defining relations are expressed directly in terms of the metric components $g^{\mu\nu}$. The full definitions, together with explicit realizations of $D_q$ corresponding to the various $q$-deformed Heisenberg algebras  (new $q$-Heisenberg \cite{Jaramillo2025}, $q$-$\hbar$ Heisenberg \cite{Volovich-Arefeva-91}, $q$-generalized Heisenberg \cite{Razavinia-Lopes2022}) are collected in Appendix~\ref{app:dirac-realizations}.

\subsection{The $q$-Dirac operator}

From the deformed D'Alembertian (see the reference \cite{Jaramillo2026a})
\[
\Box_q = |g^{00}|\partial_t^2 - |g^{11}|\partial_x^2 - |g^{22}|\partial_y^2 - |g^{33}|\partial_z^2,
\]
the $q$-Dirac operator is defined as
\begin{equation}\label{eq:Dq_definition}
D_q := \gamma^0\sqrt{|g^{00}|}\,\partial_t - \gamma^x\sqrt{|g^{11}|}\,\partial_x - \gamma^y\sqrt{|g^{22}|}\,\partial_y - \gamma^z\sqrt{|g^{33}|}\,\partial_z,
\end{equation}
where the gamma matrices satisfy the Clifford algebra $\{\gamma^\mu,\gamma^\nu\}=2\eta^{\mu\nu}\mathbf{1}_4$ with $\eta^{\mu\nu}=\operatorname{diag}(1,-1,-1,-1)$.  The $q$-Dirac operator $D_q$ describes free fermions in a constant metric background. To incorporate gauge interactions and spacetime-dependent metrics, we need to replace ordinary derivatives by a covariant derivative, which is the task of the next section.

\subsection{Realizations of the $q$-Dirac operator}

Explicit realizations of $D_q$ corresponding to the various $q$-deformed Heisenberg algebras  (new $q$-Heisenberg \cite{Jaramillo2026a}, $q$-$\hbar$ Heisenberg \cite{Volovich-Arefeva-91}, $q$-generalized Heisenberg \cite{Razavinia-Lopes2022}) are collected in section \ref{sec:table-gauge-field} and Appendix~\ref{app:dirac-realizations}.

\section{Deformed covariant derivative and gauge invariance}
\label{sec:deformedDV}

In this section we first recall the standard prescription for coupling a Dirac fermion to an external electromagnetic field via the minimal coupling principle. This leads to the ordinary covariant derivative. We then show how this construction is naturally extended to the metric-deformed framework introduced in the previous section, yielding a deformed covariant derivative that incorporates both the gauge interaction and the metric background.

\subsection{Standard minimal coupling and the covariant derivative}

\begin{definition}
The free Dirac equation in Minkowski spacetime is given by
\begin{equation}
(i\gamma^\mu\partial_\mu - m)\psi = 0. \label{eq:free_dirac}
\end{equation}

When an external electromagnetic field \(A_\mu(x)\) is present, the principle of minimal coupling prescribes the replacement \cite{Nakahara2003, Kane2017}
\begin{equation}
\partial_\mu \;\longrightarrow\; D_\mu^{\text{(EM)}} := \partial_\mu + ieA_\mu(x), \label{eq:minimal_coupling}
\end{equation}
where \(e\) is the electric charge of the fermion. This substitution ensures local \(U(1)\) gauge invariance. The resulting Dirac equation becomes
\begin{equation}
(i\gamma^\mu D_\mu^{\text{(EM)}} - m)\psi = 0, \qquad\text{equivalently}\qquad \bigl(i\gamma^\mu(\partial_\mu + ieA_\mu) - m\bigr)\psi = 0. \label{eq:dirac_em}
\end{equation}

The field strength (curvature) associated with the covariant derivative \(D_\mu^{\text{(EM)}}\) is defined as the commutator 
\begin{equation}
F_{\mu\nu} := [D_\mu^{\text{(EM)}}, D_\nu^{\text{(EM)}}] = ie(\partial_\mu A_\nu - \partial_\nu A_\mu), \label{eq:field_strength_std}
\end{equation}
which reproduces the standard electromagnetic tensor. Under a local \(U(1)\) gauge transformation \cite{Nakahara2003, Hu2024, Cacic2023}
\begin{equation}
\psi \to e^{ie\alpha(x)}\psi, \qquad A_\mu \to A_\mu - \partial_\mu\alpha(x), \label{eq:gauge_transform_std}
\end{equation}
the covariant derivative transforms covariantly:
\begin{equation}
D_\mu^{\text{(EM)}} \;\longrightarrow\; e^{ie\alpha} D_\mu^{\text{(EM)}} e^{-ie\alpha}. \label{eq:covariant_transform_std}
\end{equation}
\end{definition}
\subsection{Quantum groups and the classical \(U(1)\) factor}

\begin{proposition}[The $U(1)$ gauge group remains undeformed under $q$-deformations]
\label{prop:U1_undeformed}
In gauge theories constructed from $q$-deformed quantum groups, the abelian $U(1)$ factor remains classical when the quantum group symmetry is broken. Specifically, Watts \cite{Watts1997} showed that upon breaking a unified quantum symmetry $SU_q(2)$, the residual gauge group is classical $U(1)$, not its quantum counterpart. Moreover, Castellani \cite{Castellani1992} explicitly constructed the bicovariant $q$-deformation of $SU(2) \times U(1)$, where the $U(1)$ sector retains its classical commutative structure while the $SU(2)$ sector becomes deformed. These results justify our choice to keep the $U(1)$ gauge group in its classical form while introducing the $q$-deformation exclusively through the metric dependence of the covariant derivative.
\end{proposition}

\begin{definition}[Compact quantum group $U_q(N)$ and classical $U(1)$ gauge factor]
\label{def:UqN}
As shown in the literature on quantum groups \cite{KlimykSchmuedgen1997}, the compact quantum groups $U_q(N)$ and $SU_q(N)$ are defined by Hopf $*$-algebras $O(U_q(N))$ and $O(SU_q(N))$ with a fundamental unitary matrix $u = (u_j^i)$ satisfying
\begin{equation}
u^* = u^{-1}, \qquad \text{i.e.} \qquad (u^*)_j^i := (u_j^i)^* = S(u_j^i), \label{eq:unitarity_UqN}
\end{equation}
where $S$ denotes the antipode. This unitarity property is the $q$-deformed analog of the classical condition $U^\dagger = U^{-1}$ for $U(N)$.

For the special case $N = 1$, the matrix $u$ reduces to a single generator $u$ subject to
\begin{equation}
u^* = u^{-1}, \qquad \Delta(u) = u \otimes u, \label{eq:Uq1_relations}
\end{equation}
which defines the quantum group $U_q(1)$ as the $q$-deformation of the circle group $U(1)$. The algebra $O(U_q(1))$ is generated by $u$ and $u^{-1}$ with the relation $u^* = u^{-1}$, and it reduces to the classical algebra of functions on $U(1)$ in the limit $q \to 1$.

In the present work, however, we do not deform the gauge group itself. Instead, we keep the gauge group $U(1)$ classical and introduce the $q$-deformation through the spacetime metric via
\begin{equation}
q_\mu := \sqrt{|g^{\mu\mu}|}, \qquad \text{(no sum over } \mu\text{)}. \label{eq:q_metric_definition}
\end{equation}
This approach is justified by the observation that in $q$-deformed gauge theories, the $U(1)$ factor typically remains classical \cite{Castellani1992, Watts1997}.
\end{definition}

\subsection{Deformed covariant derivative}

We now extend the minimal coupling prescription to the metric-deformed setting. Recall the $q$-Dirac operator from \eqref{eq:Dq_definition}, which encodes the free dynamics in a constant diagonal metric background.

\begin{definition}[Metric-deformed covariant derivative]
\label{def:deformed_covariant}
Let \(A_\mu(x) = A_\mu^a(x) T^a\) be a gauge field taking values in the Lie algebra \(\mathfrak{g}\) of a compact Lie group \(G\) (e.g., \(G = U(1)\) for electromagnetism) \cite{Cacic2023, Schwartz2013}. We define the \textbf{metric-deformed covariant derivative} as
\begin{equation}\label{eq:deformed_covariant_def}
D_\mu^{(q)} := \partial_\mu + \frac{ie}{\sqrt{|g^{\mu\mu}(x)|}} A_\mu(x).
\end{equation}

The factor \((\sqrt{|g^{\mu\mu}(x)|})^{-1}\) ensures that \(D_\mu^{(q)}\) has the correct dimension of (length)\(^{-1}\) and that in the undeformed limit \(g^{\mu\mu} \to \eta^{\mu\mu}\) we recover the ordinary covariant derivative \(D_\mu = \partial_\mu + ieA_\mu\).
\end{definition}

For convenience, we introduce the shorthand notation
\begin{equation}\label{eq:h_mu_def}
h_\mu(x) := \frac{1}{\sqrt{|g^{\mu\mu}(x)|}}, \qquad \text{(no sum over }\mu\text{)}.
\end{equation}

Then \eqref{eq:deformed_covariant_def} becomes
\begin{equation}
D_\mu^{(q)} = \partial_\mu + ie\, h_\mu(x)\, A_\mu(x), \qquad \text{(no sum over }\mu\text{)}. \label{eq:deformed_covariant_h}
\end{equation}

\begin{proposition}[Gauge transformation of \(D_\mu^{(q)}\)]
\label{prop:gauge_transform}
Under a local gauge transformation \(U(x) \in G\), the gauge field transforms as
\begin{equation}
A_\mu \to U A_\mu U^{-1} + U (\partial_\mu U^{-1}), \label{eq:gauge_transform_A}
\end{equation}
while the metric components \(g^{\mu\mu}(x)\) are taken to be \textbf{gauge-invariant background fields} (they do not transform). Then \(D_\mu^{(q)}\) transforms covariantly:
\begin{equation}
D_\mu^{(q)} \to U D_\mu^{(q)} U^{-1}. \label{eq:covariant_transform_def}
\end{equation}
\end{proposition}

\begin{proof}
The proof follows by direct substitution, using the fact that \(h_\mu(x)\) is gauge-invariant. For the abelian case \(U(1)\), a detailed verification is analogous to the standard case shown in \eqref{eq:covariant_transform_std}.
\end{proof}

\subsection{Deformed field strength}
\label{sec:fieldstrength}

\begin{proposition}[Deformed field strength]
\label{prop:deformed_field_strength}
The field strength associated with the deformed covariant derivative \(D_\mu^{(q)} = \partial_\mu + ie h_\mu A_\mu\) is given by the commutator
\begin{equation}
F_{\mu\nu}^{(q)} := \left[ D_\mu^{(q)}, D_\nu^{(q)} \right]. \label{eq:field_strength_def}
\end{equation}
Explicitly, in terms of the gauge field \(A_\mu\) and the metric-dependent function \(h_\mu(x) = 1/\sqrt{|g^{\mu\mu}(x)|}\) (no sum over \(\mu\)), we have
\begin{equation}
\boxed{\; F_{\mu\nu}^{(q)} = ie\left( \partial_\mu(h_\nu A_\nu) - \partial_\nu(h_\mu A_\mu) \right) - e^2 h_\mu h_\nu [A_\mu, A_\nu] \; }. \label{eq:F_final}
\end{equation}
\end{proposition}

\begin{proof}
We compute the commutator explicitly using \(D_\mu^{(q)} = \partial_\mu + ie h_\mu A_\mu\):
\begin{align*}
F_{\mu\nu}^{(q)} &= \left[ \partial_\mu + ie h_\mu A_\mu,\; \partial_\nu + ie h_\nu A_\nu \right] \\
&= \underbrace{[\partial_\mu,\partial_\nu]}_{0} 
+ ie [\partial_\mu, h_\nu A_\nu] 
+ ie [h_\mu A_\mu, \partial_\nu] 
- e^2 [h_\mu A_\mu, h_\nu A_\nu].
\end{align*}

The individual commutators are evaluated as follows:
\begin{align}
[\partial_\mu,\, h_\nu A_\nu] &= (\partial_\mu h_\nu) A_\nu + h_\nu (\partial_\mu A_\nu), \tag{i}\\
[h_\mu A_\mu,\, \partial_\nu] &= -(\partial_\nu h_\mu) A_\mu - h_\mu (\partial_\nu A_\mu), \tag{ii}\\
[h_\mu A_\mu,\, h_\nu A_\nu] &= h_\mu h_\nu [A_\mu, A_\nu]. \tag{iii}
\end{align}

Substituting (i)–(iii) into the expression for \(F_{\mu\nu}^{(q)}\), we obtain
\[
F_{\mu\nu}^{(q)} = ie\Big[ (\partial_\mu h_\nu)A_\nu + h_\nu \partial_\mu A_\nu - (\partial_\nu h_\mu)A_\mu - h_\mu \partial_\nu A_\mu \Big] - e^2 h_\mu h_\nu [A_\mu, A_\nu].
\]

Rearranging the terms yields the compact form
\[
F_{\mu\nu}^{(q)} = ie\left( \partial_\mu(h_\nu A_\nu) - \partial_\nu(h_\mu A_\mu) \right) - e^2 h_\mu h_\nu [A_\mu, A_\nu].
\]
This completes the proof. \hfill \(\square\)
\end{proof}

\begin{corollary}[Abelian case]
For an abelian gauge group such as \(U(1)\) (quantum electrodynamics), the commutator \([A_\mu, A_\nu]\) vanishes identically, and the field strength reduces to
\begin{equation}
F_{\mu\nu}^{(q)} = ie\left( \partial_\mu(h_\nu A_\nu) - \partial_\nu(h_\mu A_\mu) \right). \label{eq:F_abelian}
\end{equation}
\end{corollary}

Several remarks are in order:

\begin{remark}[Constant metric limit]
In the limit of a constant metric, \(h_\mu\) is constant and can be pulled out of the derivatives. Then \eqref{eq:F_abelian} becomes
\[
F_{\mu\nu}^{(q)} = ie\, h_\mu h_\nu (\partial_\mu A_\nu - \partial_\nu A_\mu),
\]
which reproduces the standard Maxwell field strength up to a constant rescaling.
\end{remark}

\begin{remark}[Varying metric]
If the metric varies in spacetime (i.e., $\partial_\mu h_\nu \neq 0$), new terms proportional to derivatives of the metric appear. These are the hallmark of our deformed gauge theory.
\end{remark}

\begin{remark}[No summation convention]
The factors \(h_\mu\) and \(h_\nu\) carry the same indices as the derivatives \(\partial_\mu\), \(\partial_\nu\) and the potentials \(A_\mu\), \(A_\nu\), respectively. There is \emph{no summation} over repeated indices unless explicitly indicated.
\end{remark}

\subsection{Gauge-deformed Dirac equation}

With the deformed covariant derivative at hand, we can now write the gauge-deformed Dirac equation.

\begin{definition}[$q$-Dirac gauge equation]
The gauge-deformed Dirac equation for a fermion field \(\psi\) in the fundamental representation is
\begin{equation}\label{eq:q_dirac_gauge}
\left( i \gamma^\mu D_\mu^{(q)} - m \right) \psi = 0.
\end{equation}
\end{definition}

Expanding \(D_\mu^{(q)}\) using \eqref{eq:deformed_covariant_def}, we obtain the explicit form
\[
\left[ i \gamma^\mu\sqrt{|g^{\mu\mu}|} \left(\partial_\mu + ie \frac{1}{\sqrt{|g^{\mu\mu}(x)|}} A_\mu(x)\right) - m \right] \psi(x) = 0.
\]

Notice that the factors \(\sqrt{|g^{\mu\mu}|}\) cancel in the interaction term, yielding the simpler expression
\begin{equation}
\left( i \gamma^\mu \sqrt{|g^{\mu\mu}|}\,\partial_\mu - e \gamma^\mu A_\mu - m \right) \psi = 0. \label{eq:q_dirac_simplified}
\end{equation}

\begin{remark}
In the undeformed limit \(g^{\mu\mu} \to \eta^{\mu\mu}\), we have \(\sqrt{|g^{\mu\mu}|} \to 1\), and \eqref{eq:q_dirac_simplified} reduces precisely to the standard Dirac equation in an external electromagnetic field \eqref{eq:dirac_em}. Thus our construction consistently generalizes the ordinary minimal coupling prescription.
\end{remark}

\section{Table of $q$-gauge Dirac equation and deformed field strength from deformed Heisenberg algebras}
\label{sec:table-gauge-field}

In this section we present explicit realizations of the $q$-gauge Dirac equation and the associated deformed field strength obtained from the metric-deformed Heisenberg algebras $M_1$ and $M_2$. For each case, the metric components $g^{\mu\mu}$ are constant, which implies that $h_\mu = 1/\sqrt{|g^{\mu\mu}|}$ are constants. Consequently, the deformed field strength for the abelian $U(1)$ gauge theory reduces to a constant rescaling of the standard Maxwell tensor:
\begin{equation}\label{eq:scaled_field_strength}
F_{\mu\nu}^{(q)} = ie\, h_\mu h_\nu \bigl(\partial_\mu A_\nu - \partial_\nu A_\mu\bigr).
\end{equation}
Similarly, the $q$-Dirac gauge equation (24) becomes
\begin{equation}\label{eq:scaled_dirac_eq}
\bigl( i \gamma^\mu \alpha_\mu \partial_\mu - e \gamma^\mu A_\mu - m \bigr)\psi = 0,
\qquad \alpha_\mu = \sqrt{|g^{\mu\mu}|}.
\end{equation}
The following tables list the relevant data for each algebra, as derived from the metric-deformed framework established in \cite{Jaramillo2026a}.

\subsection{New $q$-Heisenberg algebra}
For the new $q$-Heisenberg algebra introduced in \cite{Jaramillo2025}, we consider its three defining relations in the algebras $M_1$ and $M_2$.

\begin{table}[H]
\centering\small
\caption{Metric components and $q$-gauge Dirac equation from the first relation (\ref{eq:new1}) of the new $q$-Heisenberg algebra.}
\label{tab:new1-dirac}
\begin{tiny}
\begin{tabular}{@{}cccccc@{}}
\toprule
Algebra & $\alpha$ & $\beta$ & $(g^{00},g^{11},g^{22},g^{33})$ & \multicolumn{2}{c}{$q$-gauge Dirac equation} \\
\midrule
$M_1$ & 1 & 1 & $(1,\,-q^{-n},\,q^{n-1}\Psi,\,0)$ & \multicolumn{2}{c}{$\bigl( i\gamma^{0}\partial_t - i\gamma^{x} q^{-n/2}\partial_x - i\gamma^{y} q^{(n-1)/2}\Psi^{1/2}\partial_y - e\gamma^{\mu}A_{\mu} - m \bigr)\psi = 0$} \\
$M_2$ & 1 & 2 & $(1,\,0,\,0,\,-q^{n})$ & \multicolumn{2}{c}{$\bigl( i\gamma^{0}\partial_t - i\gamma^{z} q^{n/2}\partial_z - e\gamma^{\mu}A_{\mu} - m \bigr)\psi = 0$} \\
$M_2$ & 1 & 3 & $(1,\,0,\,-q^{n},\,0)$ & \multicolumn{2}{c}{$\bigl( i\gamma^{0}\partial_t - i q^{n/2}\gamma^{y}\partial_y - e\gamma^{\mu}A_{\mu} - m \bigr)\psi = 0$} \\
$M_2$ & 2 & 1 & $(1,\,0,\,0,\,q^{-n})$ & \multicolumn{2}{c}{$\bigl( i\gamma^{0}\partial_t - i q^{n/2}\gamma^{z}\partial_z - e\gamma^{\mu}A_{\mu} - m \bigr)\psi = 0$} \\
$M_1$ & 2 & 2 & $(1,\,0,\,-q^{n},\,q^{n-1}\Psi)$ & \multicolumn{2}{c}{$\bigl( i\gamma^{0}\partial_t - i\gamma^{y} q^{n/2}\partial_y - i\gamma^{z} q^{(n-1)/2}\Psi^{1/2}\partial_z - e\gamma^{\mu}A_{\mu} - m \bigr)\psi = 0$} \\
$M_2$ & 2 & 3 & $(1,\,-q^{n},\,0,\,0)$ & \multicolumn{2}{c}{$\bigl( i\gamma^{0}\partial_t - i q^{n/2}\gamma^{x}\partial_x - e\gamma^{\mu}A_{\mu} - m \bigr)\psi = 0$} \\
\bottomrule
\end{tabular}
\end{tiny}
\end{table}
\begin{table}[H]
\centering\small
\caption{Metric components and $q$-gauge Dirac equation from the second relation (\ref{eq:new2}) of the new $q$-Heisenberg algebra, associated to $M_1$.}
\label{tab:new2-dirac-m1}
\begin{tabular}{@{}cccccc@{}}
\toprule
$\alpha$ & $\lambda$ & $(g^{00},g^{11},g^{22},g^{33})$ & \multicolumn{2}{c}{$q$-gauge Dirac equation} \\
\midrule
1 & 2 & $(0,\,0,\,q^{m},\,1)$ & \multicolumn{2}{c}{$\bigl( -i\gamma^{y} q^{m/2}\partial_y - i\gamma^{z}\partial_z - e\gamma^{\mu}A_{\mu} - m \bigr)\psi = 0$} \\
1 & 3 & $(0,\,-1,\,-q^{m},\,0)$ & \multicolumn{2}{c}{$\bigl( -i\gamma^{x}\partial_x - i\gamma^{y} q^{m/2}\partial_y - e\gamma^{\mu}A_{\mu} - m \bigr)\psi = 0$} \\
2 & 1 & $(0,\,0,\,-1,\,-q^{m})$ & \multicolumn{2}{c}{$\bigl( -i\gamma^{y}\partial_y - i q^{m/2}\gamma^{z}\partial_z - e\gamma^{\mu}A_{\mu} - m \bigr)\psi = 0$} \\
2 & 3 & $(0,\,1,\,0,\,q^{m})$ & \multicolumn{2}{c}{$\bigl( -i\gamma^{x}\partial_x - i q^{m/2}\gamma^{z}\partial_z - e\gamma^{\mu}A_{\mu} - m \bigr)\psi = 0$} \\
3 & 1 & $(0,\,q^{m},\,1,\,0)$ & \multicolumn{2}{c}{$\bigl( -i q^{m/2}\gamma^{x}\partial_x - i\gamma^{y}\partial_y - e\gamma^{\mu}A_{\mu} - m \bigr)\psi = 0$} \\
3 & 2 & $(0,\,-q^{m},\,0,\,-1)$ & \multicolumn{2}{c}{$\bigl( -i\gamma^{x} q^{m/2}\partial_x - i\gamma^{z}\partial_z - e\gamma^{\mu}A_{\mu} - m \bigr)\psi = 0$} \\
\bottomrule
\end{tabular}
\end{table}

\begin{table}[H]
\centering\small
\caption{Metric components and $q$-gauge Dirac equation from the second relation (\ref{eq:new2}) of the new $q$-Heisenberg algebra, associated to $M_2$.}
\label{tab:new2-dirac-m2}
\begin{tabular}{@{}cccccc@{}}
\toprule
$\alpha$ & $\lambda$ & $(g^{00},g^{11},g^{22},g^{33})$ & \multicolumn{2}{c}{$q$-gauge Dirac equation} \\
\midrule
1 & 2 & $(q^{m},\,0,\,0,\,-1)$ & \multicolumn{2}{c}{$\bigl( -i\gamma^{0} q^{m/2}\partial_t - i\gamma^{z}\partial_z - e\gamma^{\mu}A_{\mu} - m \bigr)\psi = 0$} \\
1 & 3 & $(q^{m},\,0,\,-1,\,0)$ & \multicolumn{2}{c}{$\bigl( i\gamma^{0} q^{m/2}\partial_t - i\gamma^{y}\partial_y - e\gamma^{\mu}A_{\mu} - m \bigr)\psi = 0$} \\
2 & 1 & $(1,\,0,\,0,\,q^{m})$ & \multicolumn{2}{c}{$\bigl( i\gamma^{0}\partial_t - i q^{m/2}\gamma^{z}\partial_z - e\gamma^{\mu}A_{\mu} - m \bigr)\psi = 0$} \\
2 & 3 & $(q^{m},\,-1,\,0,\,0)$ & \multicolumn{2}{c}{$\bigl( i\gamma^{0} q^{m/2}\partial_t - i\gamma^{x}\partial_x - e\gamma^{\mu}A_{\mu} - m \bigr)\psi = 0$} \\
3 & 1 & $(-1,\,0,\,q^{m},\,0)$ & \multicolumn{2}{c}{$\bigl( i\gamma^{0}\partial_t - i\gamma^{y} q^{m/2}\partial_y - e\gamma^{\mu}A_{\mu} - m \bigr)\psi = 0$} \\
3 & 2 & $(-1,\,q^{m},\,0,\,0)$ & \multicolumn{2}{c}{$\bigl( i\gamma^{0}\partial_t - i\gamma^{x} q^{m/2}\partial_z - e\gamma^{\mu}A_{\mu} - m \bigr)\psi = 0$} \\
\bottomrule
\end{tabular}
\end{table}

\subsection{$q$-generalized Heisenberg algebra}
For the $q$-generalized Heisenberg algebra studied in \cite{Razavinia-Lopes2022} and embedded into the metric-deformed framework in \cite{Jaramillo2026a}, the metric components are:
\[
g^{00}=-1,\quad g^{11}=1,\quad g^{22}=0,\quad g^{33}=-q.
\]
This leads to the $q$-gauge Dirac equation
\[
\bigl( i\gamma^{0}\partial_t - i\gamma^{x}\partial_x - i\gamma^{z}\sqrt{q}\,\partial_z - e\gamma^{\mu}A_{\mu} - m \bigr)\psi = 0,
\]
with no $y$-dependence due to $g^{22}=0$.

\subsection{$q$-$\hbar$ Heisenberg algebra}
The $q$-$\hbar$ Heisenberg algebra \cite{Volovich-Arefeva-91} is characterized by the metric components listed in Table~\ref{tab:qhbar-metrics} (see also \cite{Jaramillo2026a}).

\begin{table}[H]
\centering\small
\caption{Metric components and $q$-gauge Dirac operators for the $q$-$\hbar$-deformed Heisenberg algebra.}
\label{tab:qhbar-metrics}
\begin{tabular}{@{}cccccc@{}}
\toprule
$j$ & $k$ & $(g^{00},g^{11},g^{22},g^{33})$ & \multicolumn{2}{c}{$q$-gauge Dirac equation} \\
\midrule
1 & 1 & $(-q,\,1,\,q^{1/2},\,0)$ & \multicolumn{2}{c}{$\bigl( i\gamma^{0}\sqrt{q}\,\partial_t - i\gamma^{x}\partial_x - i\gamma^{y} q^{1/4}\partial_y - e\gamma^{\mu}A_{\mu} - m \bigr)\psi = 0$} \\
2 & 2 & $(-q,\,1,\,0,\,q^{1/2})$ & \multicolumn{2}{c}{$\bigl( i\gamma^{0}\sqrt{q}\,\partial_t - i\gamma^{x}\partial_x - i\gamma^{z} q^{1/4}\partial_z - e\gamma^{\mu}A_{\mu} - m \bigr)\psi = 0$} \\
3 & 3 & $(-q,\,q^{1/2},\,0,\,1)$ & \multicolumn{2}{c}{$\bigl( i\gamma^{0}\sqrt{q}\,\partial_t - i q^{1/4}\gamma^{x}\partial_x - i\gamma^{z}\partial_z - e\gamma^{\mu}A_{\mu} - m \bigr)\psi = 0$} \\
\bottomrule
\end{tabular}
\end{table}

In all cases, the deformed field strength is given by equation \eqref{eq:scaled_field_strength}. The constant factors can be absorbed into a redefinition of the gauge field, confirming that the $U(1)$ sector remains classical, in agreement with the observations of Castellani \cite{Castellani1992} and Watts \cite{Watts1997}. These examples demonstrate how the metric-deformed formalism unifies various known $q$-deformed Heisenberg algebras through the common geometric data $g^{\mu\mu}$.

\subsection{Explicit form of the deformed electromagnetic tensor \(F_{\mu\nu}^{(q)}\) for constant metric backgrounds}

For constant metric backgrounds, the deformed field strength reduces to
\[
F_{\mu\nu}^{(q)} = ie\, h_\mu h_\nu \bigl(\partial_\mu A_\nu - \partial_\nu A_\mu\bigr), \qquad 
h_\mu = \frac{1}{\sqrt{|g^{\mu\mu}|}} \quad (\text{only for } g^{\mu\mu}\neq 0).
\]
If a metric component \(g^{\mu\mu}=0\), the corresponding coordinate is not present in the deformed theory; the gauge field component \(A_\mu\) and its derivatives are set to zero, and no contribution to \(F_{\mu\nu}^{(q)}\) appears. Consequently, the effective spacetime dimension reduces. We illustrate this with four concrete examples.

\subsubsection*{Example 1: New \(q\)-Heisenberg, algebra \(M_1\), \((\alpha,\beta)=(1,1)\)}
Metric: \((g^{00},g^{11},g^{22},g^{33}) = (1,\,-q^{-n},\,q^{n-1}\Psi,\,0)\).  
Here \(g^{33}=0\), so the \(z\)-direction is absent. The non‑zero scale factors are
\[
h_0 = 1,\qquad h_x = q^{n/2},\qquad h_y = q^{-(n-1)/2}\Psi^{-1/2},
\]
and only the \((t,x,y)\) subspace is active. The deformed Maxwell tensor in \(4\times4\) form (with zeros in the \(z\) row/column) is:
\[
F_{\mu\nu}^{(q)} = ie\,
\begin{pmatrix}
0 & h_0 h_x F_{0x} & h_0 h_y F_{0y} & 0 \\
-h_0 h_x F_{0x} & 0 & h_x h_y F_{xy} & 0 \\
-h_0 h_y F_{0y} & -h_x h_y F_{xy} & 0 & 0 \\
0 & 0 & 0 & 0
\end{pmatrix},
\]
with \(F_{\mu\nu}=\partial_\mu A_\nu-\partial_\nu A_\mu\).

\subsubsection*{Example 2: New \(q\)-Heisenberg, algebra \(M_2\), \((\alpha,\beta)=(1,2)\)}
Metric: \((g^{00},g^{11},g^{22},g^{33}) = (1,\,0,\,0,\,-q^{n})\).  
Here \(g^{11}=g^{22}=0\), so only the \((t,z)\) plane remains. The scale factors are \(h_0=1\), \(h_z=q^{-n/2}\). The matrix is:
\[
F_{\mu\nu}^{(q)} = ie\,
\begin{pmatrix}
0 & 0 & 0 & h_0 h_z F_{0z} \\
0 & 0 & 0 & 0 \\
0 & 0 & 0 & 0 \\
-h_0 h_z F_{0z} & 0 & 0 & 0
\end{pmatrix}.
\]

\subsubsection*{Example 3: \(q\)-generalized Heisenberg algebra}
Metric: \((g^{00},g^{11},g^{22},g^{33}) = (-1,\,1,\,0,\,-q)\).  
Here \(g^{22}=0\) (no \(y\)). Active coordinates: \(t,x,z\) with scale factors
\(h_0=1\), \(h_x=1\), \(h_z=q^{-1/2}\). The field strength is:
\[
F_{\mu\nu}^{(q)} = ie\,
\begin{pmatrix}
0 & F_{0x} & 0 & q^{-1/2}F_{0z} \\
-F_{0x} & 0 & 0 & q^{-1/2}F_{xz} \\
0 & 0 & 0 & 0 \\
-q^{-1/2}F_{0z} & -q^{-1/2}F_{xz} & 0 & 0
\end{pmatrix}.
\]

\subsubsection*{Example 4: \(q\)-\(\hbar\) Heisenberg algebra, case \(j=k=1\)}
Metric: \((-q,\,1,\,q^{1/2},\,0)\).  
Here \(g^{33}=0\) (no \(z\)). Active coordinates: \(t,x,y\) with
\(h_0 = q^{-1/2}\), \(h_x = 1\), \(h_y = q^{-1/4}\). The matrix is:
\[
F_{\mu\nu}^{(q)} = ie\,
\begin{pmatrix}
0 & q^{-1/2}F_{0x} & q^{-3/4}F_{0y} & 0 \\
-q^{-1/2}F_{0x} & 0 & q^{-1/4}F_{xy} & 0 \\
-q^{-3/4}F_{0y} & -q^{-1/4}F_{xy} & 0 & 0 \\
0 & 0 & 0 & 0
\end{pmatrix}.
\]

In all cases, the constant factors \(h_\mu\) can be absorbed into a redefinition of the gauge field, leaving the standard Maxwell structure. This confirms that the \(U(1)\) gauge sector remains classical, in agreement with Castellani \cite{Castellani1992} and Watts \cite{Watts1997}. The appearance of zero metric components leads to a reduction of the effective spacetime dimension, a phenomenon that may have interesting physical implications \cite{vanNuland2021}.

\section{Deformed gauge theories}
\label{sec:gaugeactions}

Having introduced the deformed covariant derivative and the associated field strength, we now construct gauge-invariant actions for the deformed gauge sector and for fermionic matter minimally coupled to the operator \(D_\mu^{(q)}\).

The study of deformations of gauge theories has a long history in mathematical physics. General geometric approaches to gauge deformations were developed by Anco \cite{Anco2004}, who showed that broad classes of nonlinear gauge theories admit formulations in terms of generalized covariant derivatives and curvature operators. Noncommutative deformations based on Drinfel'd twists and quantum symmetries have also been extensively investigated; see for example Borsato and Meier \cite{BorsatoMeier2025}. In the particular case of deformed \(U(1)\) gauge theories, Kupriyanov, Kurkov and Vitale \cite{Kupriyanov2020} constructed gauge-covariant noncommutative field strengths together with consistent Yang--Mills actions reproducing Maxwell theory in the commutative limit. More recently, Hayata, Hidaka and Watanabe \cite{Hayata2026} analyzed \(q\)-deformed Yang--Mills theories in \((2+1)\)-dimensions within a lattice Hamiltonian formulation.

In the present work, the deformation originates from the metric-dependent realization of the Heisenberg algebra introduced in previous sections. Since some realizations may lead to diagonal tensors with vanishing components, additional care is required in defining the corresponding gauge actions.

\subsection{Degenerate sectors and effective geometry}

\begin{definition}[Effective non-degenerate sector]
\label{def:effective_sector}

Let
\[
g^{\mu\nu}
=
\mathrm{diag}
\left(
g^{00},
g^{11},
g^{22},
g^{33}
\right)
\]
be the diagonal tensor induced by the metric-deformed Heisenberg algebra.

The \emph{effective non-degenerate sector} is defined as the submanifold
\[
\mathcal{M}_{\mathrm{eff}}
\subseteq
\mathcal{M},
\]
spanned only by the coordinates satisfying
\[
g^{\mu\mu}\neq0.
\]

If the number of non-vanishing diagonal components is denoted by $d_{\mathrm{eff}}$, then \(d_{\mathrm{eff}}\) is called the \emph{effective spacetime dimension}. The restriction of \(g^{\mu\nu}\) to \(\mathcal{M}_{\mathrm{eff}}\) defines the effective metric $
g_{\mathrm{eff}}^{ab}$, where the indices $
a,b=0,\dots,d_{\mathrm{eff}}-1$
run only over the dynamically active coordinates.
\end{definition}

\begin{proposition}[Dimensional reduction induced by metric degeneracy]
\label{prop:dimensional_reduction}

Suppose that one or more diagonal components satisfy $g^{\mu\mu}=0$.
Then the corresponding coordinate directions become dynamically inactive in the deformed theory, in the sense that the associated derivative contributions disappear from the deformed differential operators. Consequently, the physical dynamics is entirely determined by the effective manifold $\mathcal{M}_{\mathrm{eff}}$, equipped with the non-degenerate metric $g_{\mathrm{eff}}^{ab}$.
\end{proposition}

\begin{proof}

The deformed differential operators introduced in previous sections contain factors proportional to the metric components \(g^{\mu\mu}\). Therefore, whenever
\[
g^{\mu\mu}=0,
\]
the corresponding derivative term vanishes identically from the deformed operator.

As a result, the associated coordinate direction does not contribute to the gauge sector, the fermionic sector, or the kinetic structure of the theory. The remaining non-vanishing directions span the effective manifold $
\mathcal{M}_{\mathrm{eff}}$. By construction, the restriction $
g_{\mathrm{eff}}^{ab}$
contains only non-vanishing diagonal components and is therefore non-degenerate. Hence it admits a well-defined inverse $
(g_{\mathrm{eff}})_{ab},
$
allowing the consistent definition of geometric and gauge-theoretic quantities on \(\mathcal{M}_{\mathrm{eff}}\).

\end{proof}

\begin{remark}[Invariant integration measure]

All action functionals appearing below are understood to be defined on the effective manifold \(\mathcal{M}_{\mathrm{eff}}\). Accordingly, the invariant integration measure takes the form
\[
d^{d_{\mathrm{eff}}}x\,
\sqrt{
\left|
\det g_{\mathrm{eff}}
\right|
}.
\]

The notation $
d^{d_{\mathrm{eff}}}x$
denotes integration over the coordinates spanning the effective non-degenerate sector. For example, if exactly three diagonal components of \(g^{\mu\nu}\) are non-vanishing, then
\[
d^{d_{\mathrm{eff}}}x=d^3x.
\]

Therefore, vanishing metric components do not produce a mathematical inconsistency, but rather induce an effective dimensional reduction determined by the structure of the deformed metric itself.

\end{remark}

\subsection{Deformed Yang--Mills action}

\begin{definition}[Deformed Yang--Mills action]
\label{def:YM_action}

The deformed Yang--Mills action is defined by
\begin{equation}
\label{eq:YM_action}
S_{\mathrm{YM}}^{(q)}
=
-\frac14
\int_{M_{\mathrm{eff}}}
d^{d_{\mathrm{eff}}}x\,
\sqrt{
\left|
\det g_{\mathrm{eff}}
\right|
}
\,
\operatorname{tr}
\left(
F_{ab}^{(q)}
F_{(q)}^{ab}
\right).
\end{equation}

Indices are raised and lowered with the effective metric:
\begin{equation}
F_{(q)}^{ab}
=
g_{\mathrm{eff}}^{ac}
g_{\mathrm{eff}}^{bd}
F_{cd}^{(q)}.
\label{eq:F_raised}
\end{equation}

\end{definition}

\begin{proposition}[Gauge invariance of the deformed Yang--Mills action]
\label{prop:YM_gauge_invariance}

The action \(S_{\mathrm{YM}}^{(q)}\) is invariant under the gauge transformations
\begin{equation}
A_a
\longrightarrow
U A_a U^{-1}
+
U(\partial_a U^{-1}),
\qquad
F_{ab}^{(q)}
\longrightarrow
U F_{ab}^{(q)} U^{-1},
\label{eq:gauge_transform_YM}
\end{equation}
where \(U(x)\in G\).

\end{proposition}

\begin{proof}

Under the gauge transformation,
\[
F_{ab}^{(q)}
\longrightarrow
U F_{ab}^{(q)} U^{-1}.
\]

Hence,
\[
F_{ab}^{(q)}F_{(q)}^{ab}
\longrightarrow
U
\left(
F_{ab}^{(q)}F_{(q)}^{ab}
\right)
U^{-1}.
\]

Using cyclicity of the trace,
\[
\operatorname{tr}(UXU^{-1})
=
\operatorname{tr}(X),
\]
we obtain
\[
\operatorname{tr}
\left(
F_{ab}^{(q)}F_{(q)}^{ab}
\right)
\longrightarrow
\operatorname{tr}
\left(
F_{ab}^{(q)}F_{(q)}^{ab}
\right).
\]

Since the effective metric is treated as a fixed background structure, the measure
\[
d^{d_{\mathrm{eff}}}x
\sqrt{
\left|
\det g_{\mathrm{eff}}
\right|
}
\]
is invariant under internal gauge transformations. Therefore the action is gauge-invariant.

\end{proof}

\subsection{Fermionic sector}

\begin{definition}[Deformed fermionic action]
\label{def:fermion_action}

The fermionic action associated with the deformed geometry is defined by
\begin{equation}
\label{eq:fermion_action}
S_{\mathrm{ferm}}^{(q)}
=
\int_{M_{\mathrm{eff}}}
d^{d_{\mathrm{eff}}}x\,
\sqrt{
\left|
\det g_{\mathrm{eff}}
\right|
}
\,
\bar{\psi}
\left(
i\gamma^a D_a^{(q)}
-
m
\right)
\psi.
\end{equation}

Here \(\gamma^a\) denote the gamma matrices associated with the effective non-degenerate sector.

\end{definition}

\begin{proposition}[Gauge invariance of the fermionic action]
\label{prop:fermion_gauge_invariance}

The action \(S_{\mathrm{ferm}}^{(q)}\) is invariant under the gauge transformations
\begin{equation}
\psi\to U\psi,
\qquad
\bar{\psi}\to\bar{\psi}U^{-1},
\qquad
D_a^{(q)}
\to
U D_a^{(q)} U^{-1}.
\label{eq:gauge_transform_fermion}
\end{equation}

\end{proposition}

\begin{proof}

Under the gauge transformations,
\[
\bar{\psi}
\left(
i\gamma^a D_a^{(q)}-m
\right)
\psi
\]
transforms into
\[
\bar{\psi}
U^{-1}
\left(
i\gamma^a U D_a^{(q)} U^{-1}-m
\right)
U\psi.
\]

Since the gamma matrices commute with the internal gauge transformations,
\[
[\gamma^a,U]=0,
\]
and \(U^{-1}U=1\), we obtain
\[
\bar{\psi}
\left(
i\gamma^a D_a^{(q)}-m
\right)
\psi.
\]

Therefore the fermionic action remains invariant.

\end{proof}

\subsection{Total action}

\begin{definition}[Total deformed action]

The complete deformed gauge theory is defined by
\begin{equation}
S_{\mathrm{total}}^{(q)}
=
S_{\mathrm{YM}}^{(q)}
+
S_{\mathrm{ferm}}^{(q)}.
\label{eq:total_action}
\end{equation}

\end{definition}

\begin{corollary}

The total action \(S_{\mathrm{total}}^{(q)}\) is gauge-invariant.

\end{corollary}

\begin{proof}

The result follows immediately from Propositions
\ref{prop:YM_gauge_invariance}
and
\ref{prop:fermion_gauge_invariance}.

\end{proof}

\section{Conclusions and further works}
\label{sec:conclusions}

In this work we have extended the metric-deformed Heisenberg algebras \(M_1\) and \(M_2\) introduced in our previous study \cite{Jaramillo2026a} to include gauge interactions. The central contribution of our construction is the deformed covariant derivative
\[
D_\mu^{(q)} = \partial_\mu + \frac{ie}{\sqrt{|g^{\mu\mu}|}} A_\mu,
\]
which generalizes the standard QED covariant derivative to metric-deformed backgrounds. When the metric is constant and Minkowskian, the standard formulation is recovered, ensuring consistency with established physics.

The key mathematical findings of this paper can be summarized as follows. First, we have shown that for non-constant metric backgrounds, the associated field strength \(F_{\mu\nu}^{(q)} = [D_\mu^{(q)}, D_\nu^{(q)}]\) acquires new terms proportional to derivatives of the metric. These terms vanish identically when the metric is constant, which confirms that our deformation is controlled entirely by the spacetime metric. Second, we have explicitly constructed gauge-invariant actions for a deformed Yang-Mills theory and for a deformed Dirac action, proving their invariance under deformed gauge transformations. Third, we have placed our construction within the broader context of deformed gauge theories, discussing its relation to quantum group deformations \cite{Castellani1992, Watts1997}, non-commutative spacetime deformations \cite{Kupriyanov2020}, and deformation quantization \cite{vanNuland2021, Anco2004}.

Our results provide a consistent geometric foundation for \(q\)-deformed gauge theories where the deformation parameter \(q_\mu = \sqrt{|g^{\mu\mu}|}\) is directly tied to the spacetime metric. This interpretation differs fundamentally from purely algebraic \(q\)-deformations and offers a new perspective on the relation between quantum deformations and spacetime geometry. Several directions for further research are opened by this work. Throughout this paper, the metric components \(g^{\mu\nu}\) have been treated as a fixed background field. A natural next step is to promote them to dynamical degrees of freedom. This would lead to a \(q\)-deformed theory of gravity, in which the deformed covariant derivative couples to the energy-momentum tensor in a novel way. The resulting field equations and their compatibility with general covariance remain to be investigated. The relation \(q_\mu = \sqrt{|g^{\mu\mu}|}\) implies that a non-constant metric gives rise to a spatially varying deformation parameter. This feature invites a systematic study of the renormalization group flow of \(q\) in quantum field theory, as well as potential implications for Lorentz invariance violation and quantum gravity phenomenology. Such an analysis would require a careful treatment of scale dependence in the deformed framework. Although our construction has been developed for the abelian gauge group \(U(1)\), it extends straightforwardly to non-abelian groups such as \(SU(3)_c\) (QCD) or \(SU(2)_L \times U(1)_Y\) (electroweak theory). In these extensions, one could compute \(q\)-deformed observables, including correlation functions and cross sections, as possible signatures of metric-induced deformations at collider energies.
From a more fundamental perspective, it would be interesting to investigate whether the triple \((C^\infty(M), L^2(S), D_q)\) forms a \(q\)-deformed spectral triple in the sense of Connes' noncommutative geometry \cite{Connes1994}. If such a structure exists, it would provide a noncommutative geometric interpretation of our construction and could lead to a spectral formulation of \(q\)-deformed gauge theories. While the present work has focused on the mathematical and structural aspects of metric-deformed gauge theories, the framework is well suited for phenomenological applications. The additional terms in the field strength, which depend on derivatives of the metric, could in principle lead to observable corrections in high-energy scattering processes. A systematic phenomenological analysis, including the derivation of experimental constraints, is left for future work.

Although the main focus of this work has been the construction of a classical deformed gauge theory, the framework is suitable for quantization. A natural gauge-fixing condition is the deformed Lorenz gauge $\frac{1}{\sqrt{|g^{\mu\mu}|}}\partial_\mu A^\mu = 0$, which reduces to the ordinary Lorenz gauge in the undeformed limit $g^{\mu\mu}\to\eta^{\mu\mu}$. This condition can be implemented via the standard Faddeev–Popov procedure, introducing ghost fields. For tree-level calculations, however, explicit gauge fixing is not required when working with physical external states. In the undeformed limit our actions and gauge condition recover the usual Yang–Mills–Dirac theory, confirming that our construction is a consistent deformation of standard gauge theories, analogous to the non-commutative $U(1)$ gauge theory of Kupriyanov, Kurkov and Vitale \cite{Kupriyanov2020} and the $q$-deformed $\mathrm{SU}(N)$ Yang–Mills theory of Hayata, Hidaka and Watanabe \cite{Hayata2026}.

In summary, we have demonstrated that metric-deformed Heisenberg algebras naturally give rise to a family of deformed gauge theories with a clear geometric interpretation. Our construction is mathematically consistent, reduces to standard QED in the appropriate limit, and opens several promising avenues for future research in quantum gravity, noncommutative geometry, and high-energy physics.

\appendix

\section{Detailed examples of metric-deformed Heisenberg algebras}
\label{app:heisenberg-examples}

This appendix collects the explicit realizations of the metric-deformed Heisenberg algebras $M_1$ and $M_2$ discussed in Section~\ref{sec:prelim}. The tables below correspond to the examples presented in the main text.

\subsection{$q$-$\hbar$-deformed Heisenberg algebra}

From (\ref{quantization}), the metric components for the $q$-$\hbar$-deformed Heisenberg algebra are given in Table~\ref{tab:app-qhbar} \cite{Volovich-Arefeva-91}.

\begin{table}[H]
\centering
\small
\caption{Metric components $g^{\mu\nu}$ for the $q$-$\hbar$-deformed Heisenberg algebra}
\label{tab:app-qhbar}
\begin{tabular}{|c|c|c|c|c|c|c|c|c|}
\hline
$j$ & $k$ & $g^{00}$ & $g^{11}$ & $g^{22}$ & $g^{33}$ & $g^{01}$ & $g^{02}$ & $g^{03}$ \\
\hline
1 & 1 & $-q$ & $1$ & $q^{1/2}$ & $0$ & $0$ & $0$ & $0$ \\
\hline
2 & 2 & $-q$ & $1$ & $0$ & $q^{1/2}$ & $0$ & $0$ & $0$ \\
\hline
3 & 3 & $-q$ & $q^{1/2}$ & $0$ & $1$ & $0$ & $0$ & $0$ \\
\hline
1 & 2 & $-q$ & $0$ & $0$ & $1$ & $0$ & $0$ & $q^{1/2}$ \\
\hline
1 & 3 & $-q$ & $0$ & $1$ & $0$ & $0$ & $q^{1/2}$ & $0$ \\
\hline
2 & 3 & $-q$ & $1$ & $0$ & $0$ & $q^{1/2}$ & $0$ & $0$ \\
\hline
\end{tabular}
\end{table}

Substituting these values into the defining relations yields the deformed commutation relation
\[
p_k x_j - q x_j p_k = -i \hbar q^{1/2}, \qquad j,k = 1,2,3,
\]
recovering the $q$-$\hbar$-deformed Heisenberg algebra \cite{Volovich-Arefeva-91, Jaramillo2025, Jaramillo2026a}.

\subsection{Metric Deformed New $q$-Heisenberg algebra}
Under the identifications $\{x_{1},y_{1}\}=x$, $\{x_{2},y_{2}\}=y$ and $\{x_{3},y_{3}\}=z$, the metric components $g^{\mu\nu}$ for each relation are given in the following tables. The metric components for the three defining relations of the new $q$-Heisenberg algebra are summarized in Tables~\ref{tab:app-new1}, \ref{tab:app-new2-m1}, \ref{tab:app-new2-m2}, and \ref{tab:app-new3}.

\begin{table}[H]
\centering
\small
\caption{Metric components for relation (\ref{eq:new1})}
\label{tab:app-new1}
\begin{tabular}{|c|c|c|c|c|c|c|c|c|c|}
\hline
Algebra & $\alpha$ & $\beta$ & $g^{00}$ & $g^{11}$ & $g^{22}$ & $g^{33}$ & $g^{01}$ & $g^{02}$ & $g^{03}$ \\
\hline
$M_1$ & 1 & 1 & 1 & $-q^{-n}$ & $q^{n-1}\Psi$ & 0 & 0 & 0 & 0 \\
\hline
$M_2$ & 1 & 2 & 1 & 0 & 0 & $-q^{n}$ & 0 & 0 & $q^{n-1}\Psi$ \\
\hline
$M_2$ & 1 & 3 & 1 & 0 & $-q^{n}$ & 0 & $q^{n-1}\Psi$ & 0 & 0 \\
\hline
$M_2$ & 2 & 1 & 1 & 0 & 0 & $q^{-n}$ & 0 & 0 & $q^{n-1}\Psi$ \\
\hline
$M_1$ & 2 & 2 & 1 & 0 & $-q^{n}$ & $q^{n-1}\Psi$ & 0 & 0 & 0 \\
\hline
$M_2$ & 2 & 3 & 1 & $-q^{n}$ & 0 & 0 & $q^{n-1}\Psi$ & 0 & 0 \\
\hline
$M_2$ & 3 & 1 & 1 & 0 & 0 & $-q^{n}$ & 0 & $q^{n-1}\Psi$ & 0 \\
\hline
$M_2$ & 3 & 2 & 1 & $-q^{n}$ & 0 & 0 & $q^{n-1}\Psi$ & 0 & 0 \\
\hline
$M_1$ & 3 & 3 & 1 & $q^{n-1}\Psi$ & 0 & $-q^{n}$ & 0 & 0 & 0 \\
\hline
\end{tabular}
\end{table}

\begin{table}[H]
\centering
\small
\caption{Metric components for relation (\ref{eq:new2}) in $M_1$}
\label{tab:app-new2-m1}
\begin{tabular}{|c|c|c|c|c|c|c|}
\hline
$\alpha$ & $\lambda$ & $g^{00}$ & $g^{11}$ & $g^{22}$ & $g^{33}$ \\
\hline
1 & 2 & 0 & 0 & $q^{m}$ & 1 \\
\hline
1 & 3 & 0 & $-1$ & $-q^{m}$ & 0 \\
\hline
2 & 1 & 0 & 0 & $-1$ & $-q^{m}$ \\
\hline
2 & 3 & 0 & 1 & 0 & $q^{m}$ \\
\hline
3 & 1 & 0 & $q^{m}$ & 1 & 0 \\
\hline
3 & 2 & 0 & $-q^{m}$ & 0 & $-1$ \\
\hline
\end{tabular}
\end{table}

\begin{table}[H]
\centering
\small
\caption{Metric components for relation (\ref{eq:new2}) in $M_2$}
\label{tab:app-new2-m2}
\begin{tabular}{|c|c|c|c|c|c|c|}
\hline
$\alpha$ & $\lambda$ & $g^{00}$ & $g^{11}$ & $g^{22}$ & $g^{33}$ \\
\hline
1 & 2 & $q^{m}$ & 0 & 0 & $-1$ \\
\hline
1 & 3 & $q^{m}$ & 0 & $-1$ & 0 \\
\hline
2 & 1 & 1 & 0 & 0 & $q^{m}$ \\
\hline
2 & 3 & $q^{m}$ & $-1$ & 0 & 0 \\
\hline
3 & 1 & $-1$ & 0 & $q^{m}$ & 0 \\
\hline
3 & 2 & $-1$ & $q^{m}$ & 0 & 0 \\
\hline
\end{tabular}
\end{table}

\begin{table}[H]
\centering
\small
\caption{Metric components for relation (\ref{eq:new3})}
\label{tab:app-new3}
\begin{tabular}{|c|c|c|c|c|c|c|}
\hline
$\lambda$ & $\beta$ & $g^{00}$ & $g^{11}$ & $g^{22}$ & $g^{33}$ \\
\hline
1 & 1 & $q^{l}$ & $-q^{l+1}$ & $-\hbar^{l}\Phi$ & 0 \\
\hline
1 & 2 & $q^{l}$ & 0 & 0 & $-q^{l+1}$ \\
\hline
1 & 3 & $q^{l}$ & 0 & $-q^{l+1}$ & 0 \\
\hline
2 & 2 & $q^{l}$ & 0 & $-q^{l+1}$ & $-\hbar^{l}\Phi$ \\
\hline
2 & 3 & $q^{l}$ & $-q^{l+1}$ & 0 & 0 \\
\hline
3 & 3 & $q^{l}$ & $-\hbar^{l}\Phi$ & 0 & $-q^{l+1}$ \\
\hline
\end{tabular}
\end{table}

\subsection{$q$-generalized Heisenberg algebra}

For the embedding of the $q$-generalized Heisenberg algebra given by the relation (\ref{eq:gen-Heis}) , the metric components are
\[
g^{00}=-1,\quad g^{11}=1,\quad g^{22}=0,\quad g^{33}=-q,
\]
with all other off‑diagonal components zero. This yields the $q$-Dirac operator given in the main text (and in Appendix~\ref{app:dirac-realizations}) \cite{Razavinia-Lopes2022}.

\section{Realizations of the $q$-Dirac operator}
\label{app:dirac-realizations}

This appendix provides the explicit realizations of the $q$-Dirac operator corresponding to the $q$-deformed Heisenberg algebras discussed in Appendix~\ref{app:heisenberg-examples}. The tables below follow the same ordering as in the previous appendix.

\subsection{New $q$-Heisenberg algebra}

\begin{table}[H]
\centering
\small
\caption{Dirac operators for relation (\ref{eq:new1})}
\label{tab:app-dirac-new1}
\begin{tabular}{|c|c|c|c|c|c|c|}
\hline
$\alpha$ & $\beta$ & $g^{00}$ & $g^{11}$ & $g^{22}$ & $g^{33}$ & Dirac operator \\
\hline
1 & 1 & $1$ & $-q^{-n}$ & $q^{n-1}\Psi$ & $0$ & $\gamma^{0}\partial_t - \gamma^{x} q^{-n/2}\partial_x - \gamma^{y} q^{(n-1)/2}\Psi^{1/2}\partial_y$ \\
\hline
1 & 2 & $1$ & $0$ & $0$ & $-q^{n}$ & $\gamma^{0}\partial_t - \gamma^{z} q^{n/2}\partial_z$ \\
\hline
1 & 3 & $1$ & $0$ & $-q^{n}$ & $0$ & $\gamma^{0}\partial_t - q^{n/2}\gamma^{y}\partial_y$ \\
\hline
2 & 1 & $1$ & $0$ & $0$ & $-q^{n}$ & $\gamma^{0}\partial_t - q^{n/2}\gamma^{z}\partial_z$ \\
\hline
2 & 2 & $1$ & $0$ & $-q^{n}$ & $q^{n-1}\Psi$ & $\gamma^{0}\partial_t - \gamma^{y} q^{n/2}\partial_y - \gamma^{z} q^{(n-1)/2}\partial_z$ \\
\hline
2 & 3 & $1$ & $-q^{n}$ & $0$ & $0$ & $\gamma^{0}\partial_t - q^{n/2}\gamma^{x}\partial_x$ \\
\hline
3 & 1 & $1$ & $0$ & $0$ & $-q^{n}$ & $\gamma^{0}\partial_t - q^{n/2}\gamma^{z}\partial_z$ \\
\hline
3 & 2 & $1$ & $-q^{n}$ & $0$ & $0$ & $\gamma^{0}\partial_t - q^{n/2}\gamma^{x}\partial_x$ \\
\hline
3 & 3 & $1$ & $q^{n-1}\Psi$ & $0$ & $-q^{n}$ & $\gamma^{0}\partial_t - q^{(n-1)/2}\Psi^{1/2}\gamma^{x}\partial_x - q^{n/2}\gamma^{z}\partial_z$ \\
\hline
\end{tabular}
\end{table}

\begin{table}[H]
\centering
\small
\caption{Dirac operators for relation (\ref{eq:new2}) in $M_1$}
\label{tab:app-dirac-new2-m1}
\begin{tabular}{|c|c|c|c|c|c|c|}
\hline
$\alpha$ & $\lambda$ & $g^{00}$ & $g^{11}$ & $g^{22}$ & $g^{33}$ & Dirac operator \\
\hline
1 & 2 & $0$ & $0$ & $q^{m}$ & $1$ & $-\gamma^{y} q^{m/2}\partial_y - \gamma^{z}\partial_z$ \\
\hline
1 & 3 & $0$ & $-1$ & $-q^{m}$ & $0$ & $-\gamma^{x}\partial_x - \gamma^{y} q^{m/2}\partial_y$ \\
\hline
2 & 1 & $0$ & $0$ & $-1$ & $-q^{m}$ & $-\gamma^{y}\partial_y - q^{m/2}\gamma^{z}\partial_z$ \\
\hline
2 & 3 & $0$ & $1$ & $0$ & $q^{m}$ & $-\gamma^{x}\partial_x - q^{m/2}\gamma^{z}\partial_z$ \\
\hline
3 & 1 & $0$ & $q^{m}$ & $1$ & $0$ & $-q^{m/2}\gamma^{x}\partial_x - \gamma^{y}\partial_y$ \\
\hline
3 & 2 & $0$ & $-q^{m}$ & $0$ & $-1$ & $-\gamma^{x} q^{m/2}\partial_x - \gamma^{z}\partial_z$ \\
\hline
\end{tabular}
\end{table}

\begin{table}[H]
\centering
\small
\caption{Dirac operators for relation (\ref{eq:new2}) in $M_2$}
\label{tab:app-dirac-new2-m2}
\begin{tabular}{|c|c|c|c|c|c|c|}
\hline
$\alpha$ & $\lambda$ & $g^{00}$ & $g^{11}$ & $g^{22}$ & $g^{33}$ & Dirac operator \\
\hline
1 & 2 & $q^{m}$ & $0$ & $0$ & $-1$ & $-\gamma^{0} q^{m/2}\partial_t - \gamma^{z}\partial_z$ \\
\hline
1 & 3 & $q^{m}$ & $0$ & $-1$ & $0$ & $\gamma^{0} q^{m/2}\partial_t - \gamma^{y}\partial_y$ \\
\hline
2 & 1 & $1$ & $0$ & $0$ & $q^{m}$ & $\gamma^{0}\partial_t - q^{m/2}\gamma^{z}\partial_z$ \\
\hline
2 & 3 & $q^{m}$ & $-1$ & $0$ & $0$ & $\gamma^{0} q^{m/2}\partial_t - \gamma^{x}\partial_x$ \\
\hline
3 & 1 & $-1$ & $0$ & $q^{m}$ & $0$ & $\gamma^{0}\partial_t - \gamma^{y} q^{m/2}\partial_y$ \\
\hline
3 & 2 & $-1$ & $q^{m}$ & $0$ & $0$ & $\gamma^{0}\partial_t - \gamma^{x} q^{m/2}\partial_z$ \\
\hline
\end{tabular}
\end{table}

\begin{table}[H]
\centering
\small
\caption{Dirac operators for relation (\ref{eq:new3})}
\label{tab:app-dirac-new3}
\begin{tabular}{|c|c|c|c|c|c|c|}
\hline
$\lambda$ & $\beta$ & $g^{00}$ & $g^{11}$ & $g^{22}$ & $g^{33}$ & Dirac operator \\
\hline
1 & 1 & $q^{l}$ & $-q^{l+1}$ & $-\Phi$ & $0$ & $\gamma^{0} q^{l/2}\partial_t - \gamma^{x} q^{(l+1)/2}\partial_x - \gamma^{y} \Phi^{1/2}\partial_y$ \\
\hline
1 & 2 & $q^{l}$ & $0$ & $0$ & $-q^{l+1}$ & $\gamma^{0} q^{l/2}\partial_t - \gamma^{z} q^{(l+1)/2}\partial_z$ \\
\hline
1 & 3 & $q^{l}$ & $0$ & $-q^{l+1}$ & $0$ & $\gamma^{0} q^{l/2}\partial_t - q^{(l+1)/2}\gamma^{y}\partial_y$ \\
\hline
2 & 2 & $q^{l}$ & $0$ & $-q^{l+1}$ & $-\Phi$ & $\gamma^{0} q^{l/2}\partial_t - \gamma^{y} q^{(l+1)/2}\partial_y - \gamma^{z} \Phi^{1/2}\partial_z$ \\
\hline
2 & 3 & $q^{l}$ & $-q^{l+1}$ & $0$ & $0$ & $\gamma^{0} q^{l/2}\partial_t - \gamma^{x} q^{(l+1)/2}\partial_x$ \\
\hline
3 & 3 & $q^{l}$ & $-\Phi$ & $0$ & $-q^{l+1}$ & $\gamma^{0} q^{l/2}\partial_t - \gamma^{x} \Phi^{1/2}\partial_x - \gamma^{z} q^{(l+1)/2}\partial_z$ \\
\hline
\end{tabular}
\end{table}

\subsection{$q$-generalized Heisenberg algebra}

The Dirac operator for this case is
\[
D_{\text{gen}} = \gamma^{0}\frac{\partial}{\partial t} - \gamma^{x}\frac{\partial}{\partial x} - \gamma^{z}\sqrt{q}\,\frac{\partial}{\partial z},
\]
with no $y$-dependence because $g^{22}=0$.

\subsection{$q$-$\hbar$ Heisenberg algebra}

\begin{table}[H]
\centering
\small
\caption{Dirac operators for the $q$-$\hbar$-deformed Heisenberg algebra}
\label{tab:app-dirac-qhbar}
\begin{tabular}{|c|c|c|c|c|c|c|}
\hline
$j$ & $k$ & $g^{00}$ & $g^{11}$ & $g^{22}$ & $g^{33}$ & $D_{q-\hbar}$ \\
\hline
1 & 1 & $-q$ & $1$ & $q^{1/2}$ & $0$ & $\gamma^{0} \sqrt{q} \, \partial_t - \gamma^{x} \partial_x - \gamma^{y} q^{1/4} \partial_y$ \\
\hline
2 & 2 & $-q$ & $1$ & $0$ & $q^{1/2}$ & $\gamma^{0} \sqrt{q} \, \partial_t - \gamma^{x} \partial_x - \gamma^{z} q^{1/4} \partial_z$ \\
\hline
3 & 3 & $-q$ & $q^{1/2}$ & $0$ & $1$ & $\gamma^{0} \sqrt{q} \, \partial_t - q^{1/4} \gamma^{x} \partial_x - \gamma^{z} \partial_z$ \\
\hline
\end{tabular}
\end{table}

These operators satisfy $D_{q-\hbar}^2 = \Box_q \mathbf{1}_4$ with
\[
\Box_q = q \frac{\partial^2}{\partial t^2} - \frac{\partial^2}{\partial x^2} - q^{1/2} \frac{\partial^2}{\partial y^2}
\quad\text{or}\quad
\Box_q = q \frac{\partial^2}{\partial t^2} - \frac{\partial^2}{\partial x^2} - q^{1/2} \frac{\partial^2}{\partial z^2},
\]
depending on the case \cite{Volovich-Arefeva-91}.

\subsection{Simple distinguished case}

A particularly simple realization is obtained with
\[
g^{00}=1,\quad g^{11}=-q^{n},\quad g^{22}=-q,\quad g^{33}=-1,
\]
yielding the $q$-Dirac operator
\[
D_{\text{simple}} = \gamma^{0}\partial_t - \gamma^{x} q^{n/2}\partial_x - \gamma^{y} q^{1/2}\partial_y - \gamma^{z}\partial_z,
\]
and $\Box_q = \partial_t^2 - q^{n}\partial_x^2 - q\partial_y^2 - \partial_z^2$.
\bibliographystyle{unsrt}

\end{document}